\RequirePackage{fix-cm}
\documentclass[smallextended]{svjour3}       
\smartqed  

\usepackage{amsmath,amsfonts,amssymb,ifthen,graphicx}
\usepackage{bbm}
\usepackage{algpseudocode}
\usepackage{colortbl,xcolor}
\usepackage{thmtools}
\usepackage{thm-restate}
\usepackage{url}
\usepackage{bm}
\usepackage{algorithm}
\usepackage{algorithmicx}
\usepackage[]{todonotes}

\newcommand{\Location}{\ell}
\newcommand{\Cost}{\mathrm{cost}}
\newcommand{\sCost}{\Cost^{s}}
\newcommand{\tCost}{\Cost^{t}}
\newcommand{\A}{\mathcal{A}}

\newcommand{\Space}{\mathcal{S}}
\newcommand{\ALG}{\A}

\newcommand{\OPT}{\adv{\A}}

\newcommand{\adv}[1]{{#1}^*}
\newcommand{\Round}{\pi}
\newcommand{\Phase}{\phi}

\newcommand{\Match}{\mathfrak{m}}
\newcommand{\DFunction}{\mu}
\newcommand{\Distance}{\delta}

\begin{document}

\title{Online Matching with Convex Delay Costs\thanks{A short version appears in the proceedings of ISAAC 2018 under the title Impatient Online Matching. This full version includes all proofs and new lower bounds.}}

\author{Xingwu Liu \and 
        Zhida Pan\footnote{Corresponding author} \and 
        Yuyi Wang \and 
        Roger Wattenhofer
}

\institute{X. Liu \at
           SKL Computer Architecture, ICT, CAS\\
University of Chinese Academy of Sciences, Beijing, China \\
           \email{liuxingwu@ict.ac.cn}
           \and
           Z. Pan \at
           X-Order Lab, Shanghai, China \\
           \email{zhidapan@gmail.com}
           \and
           Y. Wang \at
           X-Order Lab, Shanghai, China \\
           \email{yywang@xorder.ai}
           \and
           R. Wattenhofer \at
           ETH Z\"{u}rich \\
           \email{wattenhofer@ethz.ch}
}
\date{Received: date / Accepted: date}

\maketitle

\begin{abstract}
We investigate the problem of Min-cost Perfect Matching with Delays (MPMD) in which requests are pairwise matched in an online fashion with the objective to minimize the sum of space cost and time cost. Though linear-MPMD (i.e., time cost is linear in delay) has been thoroughly studied in the literature, it does not well model impatient requests that are common in practice. Thus, we propose convex-MPMD where time cost functions are convex, capturing the situation where time cost increases faster and faster. Since the existing algorithms for linear-MPMD are not competitive any more, we devise a new deterministic algorithm for convex-MPMD problems. For a large class of convex time cost functions, our algorithm achieves a competitive ratio of $O(k)$ on any $k$-point uniform metric space, or $O(k\Delta)$ when the metric space has aspect ratio $\Delta$. Moreover, we prove lower bounds for the competitive ratio of deterministic and randomized algorithms, indicating that our deterministic algorithm is optimal. This optimality uncover a substantial difference between convex-MPMD and linear-MPMD, since linear-MPMD allows a deterministic algorithm with constant competitive ratio on any uniform metric space. 


\end{abstract}

\section{Introduction}
\label{sec:intro}

Online matching has been studied frantically in the last years. 
Emek et al. \cite{EmekKW2016} 
started the renaissance by introducing delays and optimizing the trade-off between timeliness and quality of the matching. This new paradigm leads to the problem of Min-cost Perfect Matching with Delays (MPMD for short), where requests arrive in an online fashion and need to be matched with one another up to delays. Any solution involves two kinds of costs or penalties. One is for quality: Requests that are paired up should in some sense be as compatible as possible, so incompatibility leads to penalties. The other is for timeliness: Delay in matching a request causes a cost that is an increasing function, called the time cost function, of the waiting time. The overall objective is to minimize the sum of the two kinds of costs. 

Tractable in theory and fascinating in practice, the MPMD problem has attracted more and more attention and inspired an increasing volume of literature \cite{EmekKW2016,Emek2017Minimum,AzarCK2017,AzarCK2017arXiv,AACCGKMWW2017}. 
However, existing work in this line only studied linear time cost functions, meaning that the time penalty grows at a constant rate no matter how long the delay is. This sharply contrasts to much of our real-life experience. Just imagine a hungry customer at a restaurant: waiting a short time is no problem -- but eventually, every additional minute becomes more annoying than ever. The discontentment is experiencing convex growth, an omnipresent concept in biology, physics, engineering, or economics. 

Actually, such convex growth of discontentment appears in various real-life scenarios of online matching. For instance, online game platforms often have to match pairs of players before starting a game (consider chess as an example). Players at the same, or at least similar, level of skills should be paired up so as to make a balanced game possible. Then it would be better to delay matching a player in case of no ideal candidate of opponents. Usually it is acceptable that a player waits for a short time, but a long delay may be more and more frustrating and even make players reluctant to join the platform again. 
Another example appears in organ transplantation: 
An organ transplantation recipient may be able to wait a bit, but waiting an extended time will heavily affect its health. One may think that organ transplantation would be better modeled by bipartite matching rather than regular matching as considered in this paper; however, organ-recipients and -donors usually come in incompatible pairs that will be matched with other pairs, e.g., two-way kidney exchange
\footnote{\url{https://www.hopkinsmedicine.org/transplant/programs/kidney/incompatible/paired_kidney_exchange.html}}
. 
More real-life examples include ride sharing (match two customers), joint lease (match two roommates), just mention a few.

On this ground, we study the convex-MPMD problem, i.e., the MPMD problem with convex time cost functions. To the best of our knowledge, this is the first work on online matching with non-linear time cost. 

Convexity of the time cost poses special challenges to the MPMD problem. An important technique in solving linear-MPMD, namely, MPMD with linear time cost function, is to minimize the total costs while sacrifice some requests by possibly delaying them for a long period (see, e.g., the algorithms in \cite{AzarCK2017,Emek2017Minimum,AACCGKMWW2017}). 
Because the time cost increases at a constant rate, it is the total waiting time, rather than waiting time of individual requests, that is of interest. Hence, keeping a request waiting is not too harmful. The case of convex time costs is completely different, since we cannot afford anymore to delay old unmatched requests, as their time costs grow faster and faster. Instead, early requests should be matched early. For this reason, existing algorithms for the linear-MPMD problem do not work any more for convex-MPMD, as confirmed by examples in Section \ref{sec:alg}.

Another element of the MPMD problem is a metric space which intuitively captures how compatible the requests are. We focus on uniform metric spaces, i.e., the points in the spaces are equal-distant. This is due to the fundamental role of uniform metrics in the MPMD problem. 
A common technique in MPMD literature is to embed a general metric into a probabilistic hierarchical separated tree (HST). Since every level of an HST can be considered as a uniform metric, the task is essentially reduced to algorithm design on uniform metrics (or aspect-ratio-bounded metrics which can also be handled by our results). 
Uniform metrics are known to be tricky, e.g., Emek et al.\ \cite{Emek2017Minimum} study linear-MPMD with only two points. 
Uniform metrics also play an important role in other online problems \cite{karlin1994competitive}. For example, the $k$-server problem restricted to uniform metrics is the well-known paging problem.

In this paper, we devise a novel algorithm $\A$ for the convex-MPMD problem which is deterministic and solves the problem optimally. More importantly, our results uncover a separation: the convex-MPMD problem, even when the cost function is just a little different from linear, is strictly harder than its linear counterpart. Specifically, our main results are as follows, where $f$-MPMD stands for the MPMD problem with time cost function $f$:

\begin{restatable}{theorem}{CROfTrivialMetric}
\label{CROfTrivialMetric}
For any $f(t)=t^\alpha$ with constant $\alpha>1$, the competitive ratio of $\A$ for $f$-MPMD on $k$-point uniform metric space is $O(k)$. 
\end{restatable}
Theorem \ref{CROfTrivialMetric} can be easily generalized. On an arbitrary metric space with aspect ratio $\Delta$, the competitive ratio of $\A$ is $O(k\Delta)$.

One may wonder whether the result in Theorem \ref{CROfTrivialMetric} can be further improved because of the known result: 
\begin{theorem}[\cite{AzarCK2017,AACCGKMWW2017}]
There exists a deterministic online algorithm that solves linear-MPMD on $k$-point uniform metric space and reaches an $O(1)$ competitive ratio. 
\end{theorem}

However, for a large family of functions $f:\mathbb{R}^+\rightarrow \mathbb{R}^+$, the $f$-MPMD problem has no deterministic algorithms of competitive ratio $o(k)$. 

\begin{restatable}{theorem}{UniformMetricLowerBound}
\label{UniformMetricLowerBound}
Suppose that the time cost function $f$ is nondecreasing, unbounded, continuous and satisfies $f(0)=f'(0)=0$. 
Then any deterministic algorithm for $f$-MPMD on $k$-point uniform metric space has competitive ratio $\Omega(k)$. 
\end{restatable}

We also get a lower bound for randomized algorithms.
\begin{restatable}{theorem}{ranlowerbound}
\label{ranlowerbound}
Suppose that the time cost function $f$ is nondecreasing, unbounded, continuous and satisfies $f(0)=f'(0)=0$. 
Then any randomized algorithm for $f$-MPMD on $k$-point uniform metric space has competitive ratio $\Omega(\ln k)$ against oblivious adversaries. 
\end{restatable}

Numerous natural convex functions over the domain of nonnegative real numbers satisfy the conditions of Theorems \ref{UniformMetricLowerBound} and \ref{ranlowerbound}. Examples include monomial $f(t)=t^\alpha$ with $\alpha>1$, $f(t)=e^{\alpha t}-\alpha t-1$ with $\alpha>1$, and so on. This, together with Theorem \ref{CROfTrivialMetric}, establishes the optimality of our deterministic algorithm. Note that family of functions satisfying the conditions  is closed under multiplication and linear combination where the coefficients are positive. Hence, the lower bounds are of general significance.

\section{Related Work}\label{sec:related}

Matching has become one of the most extensively studied problems in graph theory and computer science since the seminal work of Edmonds \cite{Edmonds1965a,Edmonds1965b}. 
Karp et al.~\cite{KarpVV1990} studied the matching problem in the context of online computation which inspired a number of different versions of online matching, e.g., \cite{KalyanasundaramP1993,KhullerMV1994,MehtaSVV2005,MeyersonNP2006,BirnbaumM2008,GoelM2008,AggarwalGKM2011,DevanurJK2013,Mehta2013,Miyazaki2014,NaorW2015}. 
In these online matching problems, underlying graphs are assumed bipartite and requests of one side are given in advance. 

A matching problem where \emph{all} requests arrive in an online manner was introduced by \cite{EmekKW2016}, which  also introduced the idea that requests are allowed to be matched up to delays that need to be paid as well, 
so the problem is called Min-cost Perfect Matching with Delays (MPMD). 
The authors presented a randomized algorithm with competitive ratio $O(\log^2 k + \log \Delta)$ where $k$ is the size of the underlying metric space known before the execution and $\Delta$ is the aspect ratio of the metric space. 
Later, Azar et al.~\cite{AzarCK2017} proposed an almost-deterministic algorithm with competitive ratio $O(\log k)$. 
Ashlagi et al.\ \cite{AACCGKMWW2017} 
extended these algorithms to bipartite matching with delays (MBPMD), and achieved a competitive ratio of $O(\log k)$, and also proved a lower bound of $\Omega(\log k/\log\log k)$ for MPMD and a lower bound of $\Omega(\sqrt{\log k/\log\log k})$ for MBPMD. 
Bienkowski et al.\ \cite{bienkowski2017match} proposed the first deterministic algorithm for MBPMD problems and obtained a competitive ratio of $O\left(m^{2.46}\right)$ where $m$ is the total number of arriving requests, and this result was improved to $O\left(m\right)$ in \cite{bienkowski2018primal}. Azar et al.\ \cite{azar2018deterministic} proposed an $O\left(m^{\log (2/3+\epsilon)}\right)$-competitive deterministic algorithm, based on these two previous results, for both MPMD and MBPMD problems. 
In contrast to our work, all these papers assume that the time cost of a request is linear in its waiting time.

The idea of delaying decisions has been around for a long time in the form of rent-or-buy problems (e.g., ski rental \cite{karlin1988competitive,karlin1994competitive} and TCP acknowledgment delay \cite{dooly1998tcp,dooly2001line}), but \cite{EmekKW2016} is the first to show how to use delays in the context of combinatorial problems such as matching. 
After this, there have been a number of generalizations of classical online problems, which allow decisions with delays. To name a few, Azar et al.\ \cite{azar2017online} considered online service with delay, which generalizes the $k$-server problem. As mentioned in their paper, delay penalty functions are not restricted to be linear and even different requests can have different penalty functions. However, different delay penalty functions there do not make the problem essentially different, and there is a universal way to deal with all these functions, in contrast to the online matching problems we consider in the present paper. 
Other problems with delays include the online bin-packing \cite{azar2019price} and the online facility allocation \cite{azar2019general}.


\section{Preliminaries}
\label{sec:model}
In this section, we formulate the problem and introduce notations.
\subsection{Problem Statement}
Let $\mathbb{R}^+$ stands for the set of nonnegative real numbers. 

A metric space $\Space = (V,\DFunction)$ is a set $V$, whose members are called points, equipped with a distance function $\DFunction: V^2\rightarrow \mathbb{R}^+$ which satisfies the following conditions
\begin{description}
\item[Positive definite:] $\DFunction(x,y)\ge 0$ for any $x,y\in V$, and ``='' holds if and only if $x=y$;
\item[Symmetric:] $\DFunction(x,y)=\DFunction(y,x)$ for any $x,y\in V$;
\item[Triangle inequality
:] $\DFunction(x,y)+\DFunction(y,z)\ge\DFunction(x,z)$ for any $x,y,z\in V$.
\end{description}

Given a function $f:\mathbb{R}^+\rightarrow\mathbb{R}^+$, the problem $f$-MPMD is defined as follows, and $f$ is called the time cost function. 

For any finite metric space $\Space = (V,\DFunction)$, an online input instance over $\Space$ is a sequence $R$ of requests. Each request $\rho \in R$ is characterized by its location $\ell(\rho) \in V$ and arrival time $t(\rho) \in \mathbb{R}^+$. Assume that $|R|$ is an even number that is not pre-determined. 
The goal is to pair up the requests over time without revocation, resulting in a perfect matching of the requests. 

Specifically, suppose an algorithm $\A$ pairs up $\rho,\rho'\in R$ at time $T$. It pays the space cost $\DFunction(\Location(\rho), \Location(\rho'))$ and the time cost $f(T-t(\rho))+f(T-t(\rho'))$. The space cost of $\A$ on input $R$, denoted by $\sCost_{\A}(R)$, is the total space cost caused by all the matched pairs, and the time cost $\tCost_{\A}(R)$ is defined likewise.  The objective of the $f$-MPMD is to find an online algorithm $\A$ such that $\Cost_{\A}(R)=\sCost_{\A}(R)+\tCost_{\A}(R)$ is minimized for all $R$. 

As usual, the online algorithm $\ALG$ is evaluated by competitive analysis. Let $\OPT$ be an optimum offline algorithm\footnote{An offline algorithm knows the whole input instance at the beginning and outputs any pair $\rho,\rho'\in R$ at time $\max\{t(\rho),t(\rho')\}$.}. For any finite metric space $\Space$, if there are $a,b\in \mathbb{R}^+$ such that $\Cost_{\A}(R)\le a \cdot \Cost_{\OPT}(R)+b$ for any online input instance $R$ over $\Space$, then $\ALG$ is said to be $a$-competitive on $\Space$. The minimum such $a$ is called the competitive ratio of $\ALG$ on $\Space$. Note that both $a$ and $b$ can depend on $\Space$. If no such $a,b$ exist for some $\Space$, we say that $\ALG$ is not competitive.

This paper will focus on uniform metric spaces and monomial time cost functions $f(t)=t^\alpha, \alpha>1$. A metric space $(V,\DFunction)$ is called $\Distance$-uniform and simply denoted by $(V,\delta)$, if $\DFunction(u,v)=\Distance$ for any $u\neq v\in V$. 

\subsection{Notation and Terminology}
Any requests $\rho,\rho'$ that are paired up in the perfect matching is called a match between $\rho$ and $\rho'$ and denoted by $\langle\rho,\rho'\rangle$ or $\langle\rho',\rho\rangle$ interchangeably. A match $\langle\rho,\rho'\rangle$ is said to be \emph{external} if $\Location(\rho)\neq\Location(\rho')$, and \emph{internal} otherwise. For any request $\rho$, let $T(\rho)$ be the time when $\rho$ gets matched; $\rho$ is said to be \emph{pending} at any time $t\in (t(\rho),T(\rho))$ and \emph{active} at any time $t\in [t(\rho),T(\rho)]$. At any moment $t$, a point $v\in V$ is called \emph{aligned} if the number of pending requests at $v$ under $\mathcal{A}$ and that under $\OPT$ have the same parity,  and \emph{misaligned} otherwise. The derivative of any differentiable function $f:\mathbb{R}^+\rightarrow\mathbb{R}^+$ is denoted by $f'$.

\section{Algorithm and Analysis}\label{sec:alg}
\subsection{Basic Ideas} 

A natural idea to solve $f$-MPMD on uniform metrics is to prioritize internal matches and to create an external match only if both requests have waited long enough (say, as long as a period $\theta$). However, for any monomial time cost function $f(t)=t^\alpha, \alpha>1$, the strategy (called Strategy I) is not competitive, as illustrated in Example \ref{eg1}.

\begin{example}\label{eg1}

For any positive integer $n$ and small real number $\epsilon>0$, construct an online instance as follows. Assume that the metric space consists of two points $u$ and $v$ with distance $\Distance$.
A request $\rho_{2i}$ arrives at $u$ at time $i\cdot \theta$ for any $0\le i\le n$, while a request $\rho_{2i-1}$ arrives at $u$ at time $i\cdot\theta-\epsilon$ for any $1\le i\le n$. Point $v$ gets a request $\rho'$ at time $0$. By Strategy I, as in Figure \ref{fig:eg1}(a), each $\rho_{2i}$ is matched with $\rho_{2i+1}$ for any $0\le i< n$, and $\rho'$ and $\rho_{2n}$ are paired up, causing cost $n\cdot f(\theta-\epsilon)+f(n\theta)+\Distance$. Consider the offline solution consisting of $\langle\rho',\rho_0\rangle$ and $\langle\rho_{2i-1},\rho_{2i}\rangle$ for $1\le i\le n$, , as in Figure \ref{fig:eg1}(b), which has cost $\Distance+n\cdot f(\epsilon)$. 
When $n$ approaches infinity and $\epsilon$ approaches 0, $n\cdot f(\theta-\epsilon)+f(n\theta)+\Distance\gg \Distance+n\cdot f(\epsilon)$, meaning that Strategy I is not competitive.
\end{example}
\begin{figure}
  \centering
  \includegraphics[width=\textwidth]{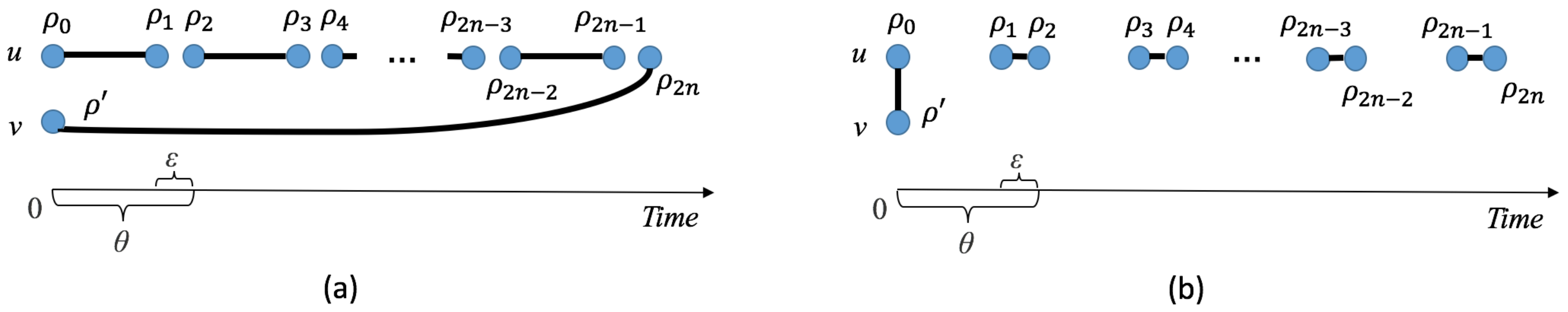}
  \caption{The input instance of Example \ref{eg1}. A blue dot stands for a request, and a thick line or curve for a match. (a) is the matching produced by Strategy I, while (b) is an offline solution.}\label{fig:eg1}
\end{figure}

Can we improve Strategy I so as to beat Example \ref{eg1}? A plausible way is as follows. At each point $v$ in the metric space, accumulate from scratch the time costs of all the requests which have arrived at $v$ since the last external match involving $v$. An external match between point $u,v$ is enable only if both have accumulated enough time costs (say, as large as $\theta$). Though beating Example \ref{eg1}, this improvement (called Strategy II) remains not competitive for any time cost function $f(t)=t^\alpha, \alpha>1$, as shown in the next example. 


\begin{example}\label{eg2}
Arbitrarily fix an even integer $n>0$ and a small real number $\epsilon>0$. Let $\tau\in\mathbb{R}^+$ be such that $\theta-\epsilon<\frac{n}{2} f(\tau)< \theta$. Again, assume that the metric space consists of two points $u$ and $v$ with distance $\Distance$. Suppose that a request $\rho'$ arrives at $v$ at time 0, while a request $\rho_i$ arrives at $u$ at time $i\tau$ for any $0\le i\le n$. Hence there are totally $n+2$ requests. As illustrated in Figure \ref{fig:eg2}(a), applying Strategy II results in the matches $\langle \rho',\rho_n\rangle$ and $\langle \rho_i,\rho_{i+1}\rangle$ for any even number $0\le i<n$, causing cost at least $\frac{n}{2} f(\tau)+f(n\tau)+\Distance$. On the other hand, consider the offline solution $\langle \rho',\rho_0\rangle$ and $\langle \rho_i,\rho_{i+1}\rangle$ for any odd number $0< i<n$, as shown in Figure \ref{fig:eg2}(b). It has cost $\frac{n}{2} f(\tau)+\Distance$. Thus the cost of $\OPT$ is at most $\frac{n}{2} f(\tau)+\Distance$. When $n$ approaches infinity and $\epsilon$ approaches 0, we have $\frac{n}{2} f(\tau)+f(n\tau)+\Distance\gg \frac{n}{2} f(\tau)+\Distance$, implying that Strategy II is not competitive.
\end{example} 
\begin{figure}
  \centering
  \includegraphics[width=\textwidth]{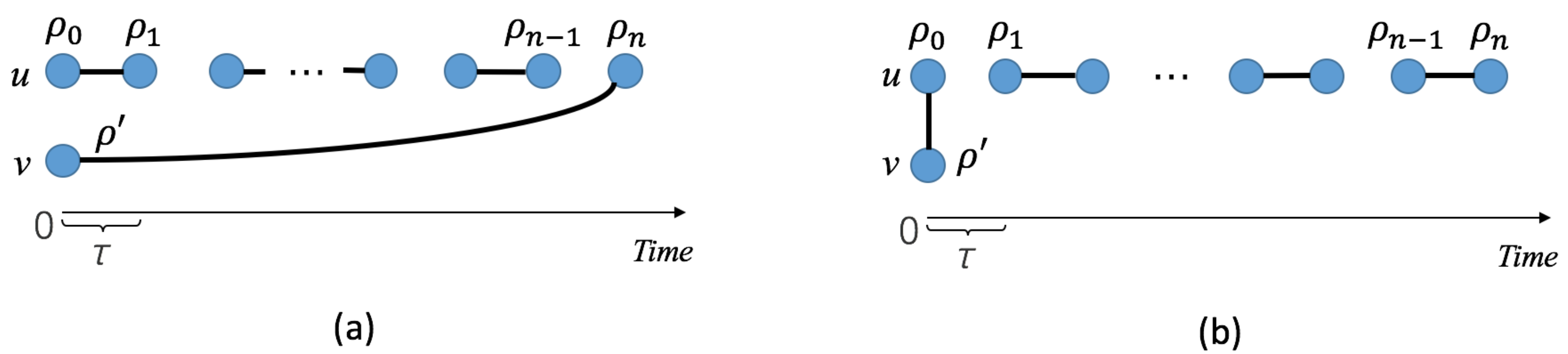}
  \caption{The input instance of Example \ref{eg2}.  A blue dot stands for a request, and a thick line or curve for a match. (a) is the matching produced by Strategy II, while (b) is an offline solution.}\label{fig:eg2}
\end{figure}

The failure in Example \ref{eg2} is rooted at the \emph{double-enabling} mechanism which enables an external match only if the cumulated costs at both points reach some threshold. 
Hence, we replace that mechanism with a looser one: an external match is enabled if one of the two points has high accumulated cost (say, as high as $\theta$). The new algorithm (called Strategy III) defeats both Examples \ref{eg1} and \ref{eg2}, but the following example shows that it remains not competitive for any monomial time cost function $f(t)=t^\alpha, \alpha>1$. 

\begin{example}\label{eg3}
Choose $\tau\in\mathbb{R}^+$ and odd integer $n>0$ such that $f(n\tau)=\theta$. Arbitrarily fix real number $T_0>n\tau$. Consider a $\Distance$-uniform metric space with points $u,v,w$. Let $m>0$ be an arbitrary integer. Construct an online input instance $R$ which is the catenation of $m+1$ parts $R_0,\cdots,R_m$, as illustrated in Figure \ref{fig:eg3}.

\begin{figure}
  \centering
  \includegraphics[width=\textwidth]{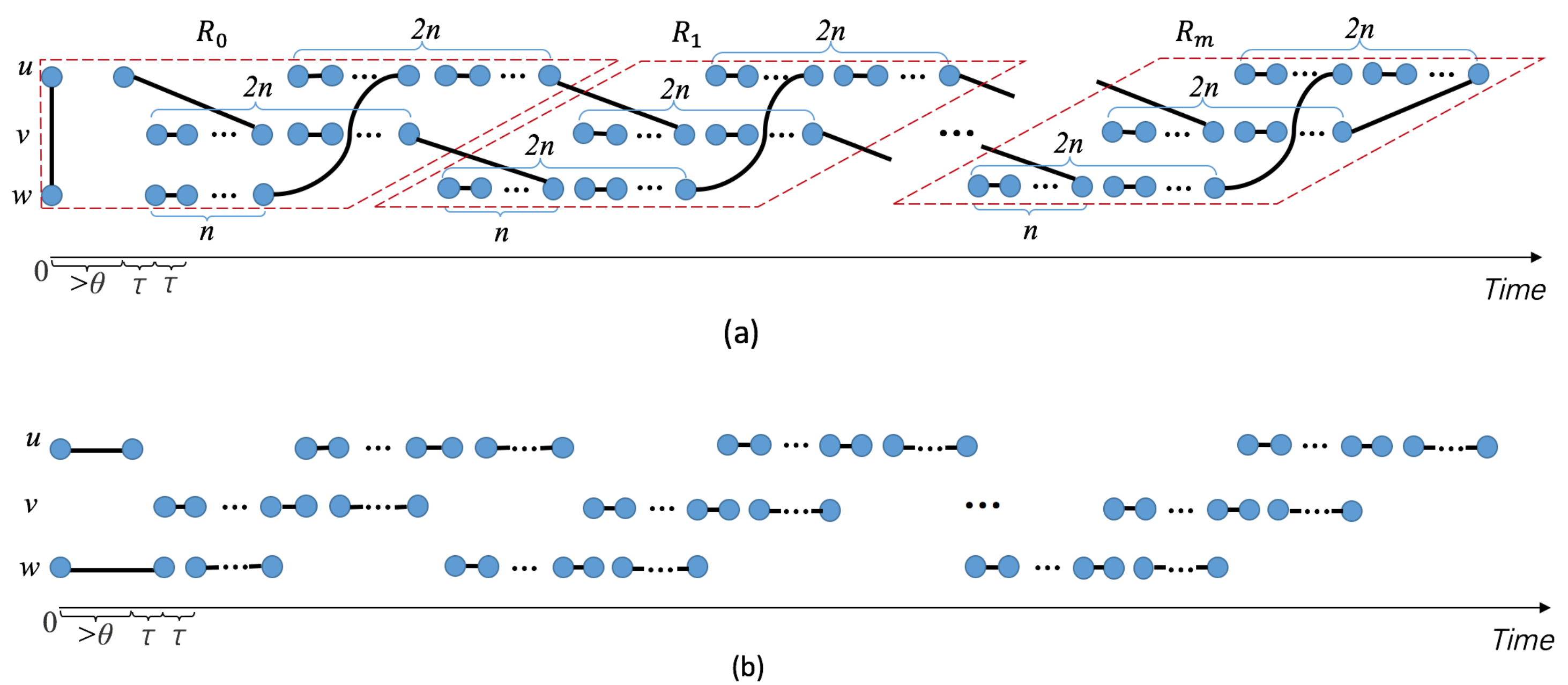}
  \caption{The input instance of Example \ref{eg3}. A blue dot stands for a request, an area surrounded by dash lines stands for a part of the instance, and a thick line or curve for a match.  (a) The matching produced by Strategy III. (b) An offline solution. \textcolor{blue}{Label $T_i$ is (a)}}\label{fig:eg3}
\end{figure}

The part $R_0$ has $5n+3$ requests. Specifically, $u$ receives a request $\rho^u_{0,-1}$ at time $0$, $\rho^u_{0,0}$ at time $T_0$, and $\rho^u_{0,i}$ at time $T_0+(n+i)\tau$ for any $1\le i\le 2n$. $v$ receives a request $\rho^v_{0,i}$ at time $T_0+i\tau$ for any $1\le i\le 2n$. $w$ receives a request $\rho^w_{0,-1}$ at time $0$ and a request $\rho^w_{0,{n+i}}$ at time $T_0+i\tau$ for any $1\le i\le n$. Let $T_1=T_0+(2n+1)\tau, T_j=T_{j-1}+3n\tau$ for any $2\le j\le m$. 

For any $1\le j\le m$, the part $R_j$ has $6n$ requests as follows: $\rho^u_{j,i}$ arrives at $u$ at time $T_j+(2n+i-1)\tau$, $\rho^v_{j,i}$ arrives at $v$ at time $T_j+(n+i-1)\tau$, and $\rho^w_{j,i}$ arrives at $w$ at time $T_j+(i-1)\tau$, for every $1\le i\le 2n$. 

Actually, we can slightly perturb the arrival time of some requests so that Strategy III results in exactly the following external matches: $\langle\rho^u_{0,-1},\rho^w_{0,-1}\rangle$, $\langle\rho^u_{0,0},\rho^v_{0,n}\rangle$, $\langle\rho^u_{j,n},\rho^w_{j,2n}\rangle$ for $1\le j\le m$, $\langle\rho^u_{i,2n},\rho^v_{i+1,n}\rangle$ and $\langle\rho^v_{i,2n},\rho^w_{i+1,n}\rangle$ for $1\le i< m$, and $\langle\rho^u_{m,2n},\rho^v_{m,2n}\rangle$, as illustrated in Figure \ref{fig:eg3}(a). The cost of Strategy III is at least $3m(\Distance+\theta)$. On the other hand, consider the offline solution SOL which has no external matches, as indicated in Figure \ref{fig:eg3}(b). It has cost at most $2f(T_0+\tau)+\frac{6mn+5n-1}{2}f(\tau)$. When $\tau$ approaches zero and $m$ approaches infinity, we have $3m(\Distance+\theta)\gg 2f(T_0+\tau)+\frac{6mn+5n-1}{2}f(\tau)$, implying that Strategy III is not competitive.
\end{example}

Let's look closer at the example. Consider an arbitrary (except the first) external match $\langle\rho,\rho'\rangle$ of Strategy III. It is of misaligned-aligned pattern in the sense that $\Location(\rho)$ and $\Location(\rho')$ have opposite alignment status when the match occurs. Suppose $\Location(\rho)$ is misaligned. Then it has accumulated high cost, mainly due to the long delay of $\rho$. On the contrary, SOL has accumulated little cost at $\Location(\rho)$, because SOL has no pending request there while $\rho$ is pending. Hence, a match of misaligned-aligned pattern can significantly enlarge the gap between online/offline costs. To be worse, such a match does not change the number of aligned/misaligned points, making it  possible that this pattern appears again and again, enlarging the gap infinitely. As a result, we establish a set which consists of points that are likely to be misaligned, and prioritize matching those requests that are located outside the set. The algorithm is described in detail as follows.

\subsection{Algorithm Description}\label{algorithm}
Our algorithm maintains a subset $\Psi \subseteq V$ to store points  tend to be misaligned. For every point $v\in V$, we use a timer $z_v\in \mathbb{R}^+$ to accumulate the time cost at $v$. The timer $z_v$ is initialized to be $0$, and increases at rate $f'(t-t_0)$ at time $t$ when there is an active request $\rho$ with $l(\rho)=v$ and $t(\rho)=t_0$. 
The algorithm proceeds round by round, and $\Psi$ is reset to be the empty set $\emptyset$ at the beginning of each round. 
The first round begins when the algorithm starts. 
Let $k=|V|$. Whenever $2k$ external matches are produced, the present round ends immediately and the next one begins. At any time $t$, the following operations are performed exhaustively, i.e., until there is no possible matching according to the rules.
\begin{enumerate}
\item Pair up active requests $\rho$ and $\rho'$ if  $\Location(\rho)=\Location(\rho')$.
\label{line2}
\item For active requests $\rho,\rho'$ with $u\triangleq\Location(\rho)\neq v\triangleq \Location(\rho')$, if there is $x\in\{u,v\}$ satisfying either (a) $\Distance\leq z_x< 2\Distance$ and  $\{u,v\}\bigcap \Psi=\emptyset$ or (b) $z_x\geq 2\Distance$, perform the three actions:
      \begin{itemize}
	\item Pair up $\rho$ with $\rho'$.
	\item Reset $z_{u}=z_{v}=0$.
	\item Set $\Psi$ to be $(\Psi\setminus \{u,v\})\bigcup\{x\}$ if either $u\notin \Psi$ or $v\notin \Psi$.
     \end{itemize}
     In this case, we say that $x$ initiates the match $\langle\rho,\rho'\rangle$. \label{externalmatch}
\end{enumerate}
Priority rule: in applying Operation \ref{externalmatch}, the requests located outside $\Psi$ are prioritized.

\subsection{Competitive Analysis}
Throughout this subsection, arbitrarily fix a time cost function $f(t)=t^\alpha$ with $\alpha>1$, a uniform metric space $\Space=(V,\Distance)$ of $k$ points, and an arbitrary online input instance $R$ over $\Space$. For ease of presentation, we assume that the arrival times of the requests are pairwise different. This assumption does not lose generality since the arrival times might be slightly perturbed and timing in practice is up to errors. Let $\ALG$ stand for our algorithm and $\OPT$ for an optimum offline algorithm solving $f$-MPMD. We start competitive analysis by introducing notation.
\subsubsection{Notation}
For any request $\rho\in R$ and subset $I\subseteq \mathbb{R}^+$ of time, the time cost of $\OPT$ incurred by $\rho$ during $I$ is defined to be $$C_{time}(\rho,I,\OPT)=
\int_{(t(\rho),T^*(\rho)]\bigcap I} f'(t-t(\rho))dt,$$ where $T^*(\rho)$ is the time when $\rho$ gets matched by $\OPT$. For any $v\in V$, define $$C_{time}(v,I,\OPT)=\sum_{\rho\in R,\Location(\rho)=v}C_{time}(\rho,I,\OPT).$$ Let $C_{space}(v,I,\OPT)$ be $\frac{\Distance}{2}$ times the number of requests at $v$ that are externally matched by $\OPT$ during $I$.

Define $$\Gamma=\{t\in\mathbb{R}^+: \textrm{at time } t, \mathcal{A} \textrm{ has a pending request } \rho \textrm{  with }z_{\Location(\rho)}>2\Distance\}.$$ We will analyse time cost of $\OPT$ inside and outside $\Gamma$ separately.

Our algorithm $\ALG$ runs round by round as described in Section \ref{algorithm}. Specifically, the \emph{round} starting at time $t_0$ and ending at time $t_1$ corresponds to the time period $(t_0,t_1]$. Let $\Pi$ be the set of rounds of $\ALG$. 

For any  $\Round\in\Pi$, define $$round\_cost_{time}(\pi,\OPT)=\sum_{v\in V}C_{time}(v,\pi\setminus \Gamma,\OPT)$$ which stands for the time cost of $\OPT$ during $\pi\setminus \Gamma$, and $$round\_cost_{space}(\pi,\OPT)=\sum_{v\in V}C_{space}(v,\pi,\OPT)$$ which is the space cost of $\OPT$ during $\pi$. 

Arbitrary fix a round $\pi$ and an external match $\Match=\langle \rho, \rho \rangle$ initiated by $l(\rho)\in V$. Define the phase of $\Match$, denoted by $\phi(\Match)$,  to be the period $(t,T(\rho)]$, where $t<T(\rho)$ is the time closest to $T(\rho)$ such that $z_{l(\rho)}=0$ at time $t$. Define the value of $\Match$, denoted by $\sigma(\Match)$, to be the value of $z_{l(\rho)}$ at time $T(\rho)$. If $T(\rho)\in \pi$, we say that $\Match$ is a match of round $\pi$. If $\phi(\Match)\subseteq \pi$, $\Match$ is called a native match of $\pi$. $\Match$ is called a good match, if the alignment status (namely, being aligned or misaligned) of $l(\rho)$ never change throughout $\phi(\Match)$. If all the external matches of $\pi$ are good, $\pi$ is called a \emph{good} round. Round $\pi$ is called \emph{complete} if it has $2k$ external matches. 

\subsubsection{Competitive Ratio of Our Algorithm}
\label{CompetitiveRatio}

Basically, we show that in every round, the increment of the cost of $\ALG$ does not differ too much from that of $\OPT$ . This is reduced into two tasks. First, if the timers are always small (say, no more than 4$\Distance$), the cost of $\ALG$ increases by $O(kd)$ in every round, while that of $\OPT$ increases by $\Omega(d)$. This is the main task of this subsection and presented in Lemma \ref{lemma: at most n exchanges}. Second, in case that some timer $z_v$ gets too large, we show that the cost of $\OPT$  increases simultaneously and almost proportionately, as claimed in Lemma \ref{totalcostlowerbound}. 

Let's begin with some technical lemmas that will be frequently used. 

\begin{lemma}\label{lemma:convexproperty}
Let $h:\mathbb{R}^+\mapsto\mathbb{R}^+$ be an invertible increasing convex function. The inequality $h(h^{-1}(\xi)-h^{-1}(\eta))+\zeta \ge h(h^{-1}(\xi+\zeta)-h^{-1}(\eta))$ holds for any $\xi,\eta,\zeta\in\mathbb{R}^+$ with $\xi\geq\eta$.
\end{lemma}
\begin{proof}
Let $x=h^{-1}(\xi), y=h^{-1}(\eta), z=h^{-1}(\xi+\zeta)$. Note that $y\le x\le z$. Then $h(z)-h(z-y)=\int_{(z-y,z]}h'(t)dt=\int_{(x-y,x]}h'(t+z-x)dt$. By convexity of $h$, $h'$ is increasing, implying that $h(z)-h(z-y)\ge \int_{(x-y,x]}h'(t)dt=h(x)-h(x-y)$. As a result, $h(x-y)+h(z)-h(x)\ge h(z-y)$, which is exactly the desired inequality. \qed
\end{proof}

\begin{lemma}\label{lemma:alignedlowerbound}
Suppose that $\rho_1,\cdots,\rho_n\in R$ with $T(\rho_i)<t(\rho_{i+1})$ for any $1\le i<n$ are successive pending requests at $v\in V$. 
Let $\gamma$ and $\lambda$ be the value of $z_v$ at some time  $t_1\in (t(\rho_1),T(\rho_1)]$ and $T_n\in (t(\rho_n),T(\rho_n)]$, respectively. Let $t_i=t(\rho_i)$ for $1<i\le n$ and $T_j=T(\rho_j)$ for $1\le j< n$. Then $\sum_{i=1}^n f(T_i-t_i)\ge f(f^{-1}(\lambda)-f^{-1}(\gamma))$.
\end{lemma}
\begin{proof}
For any $1\le i\le n$, let $c_i$ be the increment of $z_v$ during $I_i=(t_i,T_i]$, i.e. $c_i\triangleq\int_{I_i} f'(t-t(\rho_i))dt$. Then we have $\lambda-\gamma\le\sum_{i=1}^n c_i$.

When $i>1$, $c_i=f(T_i-t_i)$ because $t(\rho_i)=t_i$. 

Now it comes to $i=1$. Since $z_v=\gamma$ at time $t_1$, $f(t_1-t(\rho_1))=\int_{(t(\rho_1),t_1]}f'(t-t(\rho_1))dt\le \gamma$.
Because $c_1=\int_{(t_1,T_1]}f'(t-t(\rho_1))dt=f(T_1-t(\rho_1))-f(t_1-t(\rho_1))$, it holds that $T_1-t_1=f^{-1}(c_1+x)-f^{-1}(x)$ where $x=f(t_1-t(\rho_1))$. By convexity of $f$ and $x\le \gamma$, we have $f^{-1}(c_1+x)-f^{-1}(x)\ge f^{-1}(c_1+\gamma)-f^{-1}(\gamma)$. Then, 
\begin{align*}
\sum_{i=1}^n f(T_i-t_i)&\ge f(f^{-1}(c_1+\gamma)-f^{-1}(\gamma))+c_2+\cdots+c_n\\
                                   &\ge f(f^{-1}(\gamma+c_1+c_2+\cdots+c_n)-f^{-1}(\gamma))\\
                                   &\ge f(f^{-1}(\lambda)-f^{-1}(\gamma)),
\end{align*}
where the second inequality follows from Lemma \ref{lemma:convexproperty}. \qed
\end{proof}
\begin{lemma}\label{lemma:aligned initiation}
If an aligned point $v\in V$ initiates a good native match $\Match$ of round $\pi$, then $round\_cost_{time}(\Round,\OPT) \ge \min\{\sigma(\Match),2\Distance\}$.
\end{lemma}
\begin{proof}
Let $\rho_1,\cdots,\rho_n\in R$ with $T(\rho_i)<t(\rho_{i+1})$ for any $1\le i<n$ be the requests at $v$ that are successively pending during $\phi(\Match)$. We proceed case by case.

Case 1: $\sigma(\Match)\le 2\Distance$. Since $v$ is aligned throughout $\phi(\Match)$, $\OPT$ has requests $\rho'_1,\cdots,\rho'_n\in R$ at $v$ with $t(\rho'_i)\le t(\rho_i)$ and $T(\rho'_i)\ge T(\rho_i)$ for any $1\le i\le n$. 
Then by Lemma \ref{lemma:alignedlowerbound}, $round\_cost_{time}(\Round,\OPT)\ge \sum_{i\ge 1}^n f(T(\rho_i)-t(\rho_i))=\sigma(\Match)$.

Case 2: $\sigma(\Match)> 2\Distance$. Suppose $z_v$ reaches $2\Distance$ exactly at time $\widehat{T}$ when $\rho_m$ is active. Adopting the notation in Case 1, we have $round\_cost_{time}(\Round,\OPT)\ge \sum_{i\ge 1}^{m-1} f(T(\rho_i)-t(\rho_i))+f(\widehat{T}-t(\rho_m))=2\Distance$. \qed
\end{proof}
%

\begin{lemma}\label{lemma:no misaligned points in Psi}
If a misaligned point appears in $\Psi$ in a good round $\Round$, then we have $round\_cost_{time}(\pi,\OPT) \ge \Distance$ or $round\_cost_{space}(\pi,\OPT) \ge \Distance$.
\end{lemma}
\begin{proof}
Let $v$ be the first misaligned point in $\Psi$ during the round $\Round$, namely, during $\Round$, any point in $\Psi$  is aligned before $v\in\Psi$ gets misaligned. Then we proceed case by case.

\textbf{Case 1}:  $v$ is misaligned when it goes into $\Psi$. By the rule of updating $\Psi$, $v$ goes into $\Psi$ due to an external match $\Match$ initiated by $v$. Hence, before $\Match$ occurs, $v$ is aligned. Then $round\_cost_{time}(\pi,\OPT) \ge \Distance$ by Lemma \ref{lemma:aligned initiation}.

\textbf{Case 2}:  $v$ is aligned when it goes into $\Psi$, but gets misaligned due to an external match of $\OPT$. Then $round\_cost_{space}(\pi,\OPT) \ge \Distance$.

\textbf{Case 3}:  $v$ is aligned when it goes into $\Psi$, but gets misaligned due to an external match $\Match$ of $\ALG$. Then before $\Match$ occurs, $v$ is aligned. Again by the rule of updating $\Psi$, $\Match$ must be initiated either by $v$ or by another point $u\in \Psi$. Anyway, the initiating point must be aligned before $\Match$ occurs, since $v$ is the first misaligned point in $\Psi$ during this round. As a result, $round\_cost_{time}(\pi,\OPT) \ge \Distance$ by Lemma \ref{lemma:aligned initiation}. \qed
\end{proof}

Roughly speaking, the next lemma claims that under some condition, even if $\ALG$ pairs up one request located inside $\Psi$ with another outside $\Psi$, the cost of $\OPT$ must increase substantially.
\begin{lemma}\label{lemma:phase included in round}
Suppose in a good round $\Round$, $v\notin \Psi$ initiates a native match $\Match$ between a request at $v$ and another at $v'\in \Psi$. Then one of the following two claims is true:
\begin{itemize}
\item $round\_cost_{time}(\pi,\OPT)\ge f(f^{-1}(2\Distance)-f^{-1}(\Distance))$,
\item $round\_cost_{space}(\pi,\OPT)\ge \Distance$.
\end{itemize}
\end{lemma}
Basic idea of the proof: Since $\Match$ is between $v\notin \Psi$ and $v'\in \Psi$ and initiated by $v$, it holds that $z_v\ge 2\Distance$ when $\Match$ occurs. All we need to prove is that in the process that $z_v$ increases from $\Distance$ to $2\Distance$, whenever $\ALG$ has a pending request $\rho$ at $v$, $\OPT$ also has a request that stays pending at least as long as $\rho$ does. Then the proof ends due to Lemma \ref{lemma:alignedlowerbound}. 
\begin{proof}
If there exists a misaligned point in $\Psi$ during $\Round$, according to Lemma \ref{lemma:no misaligned points in Psi}, the assertion follows. 
If $v$ is aligned throughout the period $\Phase(\Match)$, according to Lemma \ref{lemma:aligned initiation}, the assertion again follows.

Now we focus on the other case, namely, all points in $\Psi$ are aligned, and $v$ is misaligned in $\Phase(\Match)$.

Let $\rho_1,...,\rho_n$ with $T(\rho_i)\le t(\rho_{i+1})$ for each $i$ be the pending requests at $v$ that cause $z_v$ to increase from $\Distance$ to $2\Distance$. Choose $t(\rho_{1})\le a_1<T(\rho_{1})$ and $t(\rho_{n})<b_n\le T(\rho_{n})$ such that $z_v=\Distance$ at time $a_1$ and $z_v=2\Distance$ at time $b_n$. 
Let $a_i=t(\rho_{i})$ for any $1<i\le n$, $b_i=T(\rho_{i})$ for any $1\le i< n$, and $I_i=(a_i,b_i]$ for any $1\le i\le n$. Then $\sum_{i=1}^n \int_{I_i} f'(t-t(\rho_i))dt=2\Distance-\Distance=\Distance$. 

Now we have three observations.
\begin{enumerate}
\item \emph{During each time interval $I_i$, no request is pending outside $\Psi\bigcup\{v\}$.} \\For contradiction, suppose there is a pending request $\rho'$ in $I_i$ with $l(\rho')\notin\Psi\bigcup\{v\}$. Since $z_v\ge \Distance$ and there is a pending request $\rho$ at $v$ during $I_i$, $\ALG$ should match $\rho$ with $\rho'$, contradictory to the assumption that $\Match$ is between requests at $\Psi$ and $V\setminus \Psi$. \label{no_pending}
\item \emph{During each time interval $I_i$, no requests arrive at any point outside $\Psi\bigcup\{v\}$.} \\Suppose on the contrary that a request $\rho'$ arrives at $u\notin\Psi\bigcup\{v\}$ during $I_i$. By Observation \ref{no_pending}, only $v$ has a pending request outside $\Psi$, which must get matched with $\rho'$ due to the priority rule. This means that $\Match$ occurs outside $\Psi$, again contradictory to the assumption of the lemma. \label{no arrival}
\item \emph{During each time interval $I_i$, $\Psi$ remains unchanged.} \\First, we argue that no point is added to $\Psi$. Suppose on the contrary that some $u$ is added to $\Psi$ during $I_i$. This means that an external match $\Match'=\langle\rho,\rho'\rangle$ initiated by $u$ occurs during $I_i$. Without loss of generality, assume $u=\Location(\rho), w=\Location(\rho')$. Since at any moment at most one request arrives, either $\rho$ or $\rho'$ arrives before $\Match'$ occurs. By Observation \ref{no_pending}, the only possibility is that $\rho'$ arrives before $\Match'$ occurs and $w\in \Psi$, which contradicts the priority rule of $\ALG$. 

Second, we show that no point is removed from $\Psi$. Suppose on the contrary that some $u$ is removed from $\Psi$ during $I_i$. Since no point is added to $\Psi$ during $I_i$, the size of $\Psi$ decreases by one when $u$ is removed, which is contradictory to the rule of updating $\Psi$.
\end{enumerate}

Since the number of misaligned points is even, at any moment during $\bigcup_{i=1}^n I_i$ when $v$ is misaligned, there must be a misaligned point outside $\Psi\bigcup\{v\}$. By the above observations and the definition of alignment status, for any $1\le i\le n$, $\OPT$ must have a request $\rho'_i$ that is pending throughout $I_i$. Let $u_i=\Location(\rho'_i)$ for any $1\le i\le n$.

Since each $\rho'_i$ is pending throughout $I_i$ and $f'$ is increasing, $$C_{time}(u_i,I_i,\OPT)\ge\int_{I_i}f'(t-t(\rho'_i))dt\ge \int_{I_i}f'(t-a_i)dt=f(b_i-a_i).$$

As a result,
\begin{align*}
round\_cost_{time}(\pi,\OPT)&\ge \sum_{i=1}^n C_{time}(u_i,I_i,\OPT)\\
			                    &\ge \sum_{i=1}^n f(b_i-a_i)\\
			                    &\ge f(f^{-1}(2\Distance)-f^{-1}(\Distance)),
\end{align*}
where the last inequality follows from Lemma \ref{lemma:alignedlowerbound}. \qed
\end{proof}

It is time to prove the following key lemma, stating that in every good complete round of $\ALG$, the cost of the optimum offline algorithm $\OPT$ is not small.

\begin{lemma}\label{lemma: at most n exchanges}
In any good complete round $\Round$, one of the following three claims must be true:
\begin{itemize}
\item $round\_cost_{time}(\pi,\OPT)\ge f(f^{-1}(2\Distance)-f^{-1}(\Distance))$,
\item $round\_cost_{space}(\pi,\OPT)\ge \Distance$.
\end{itemize}
\end{lemma}
\begin{proof}
Let $\mathfrak{M}$ be the set of external matches $\ALG$ outputs during $\Round$. By definition, $|\mathfrak{M}|=2k$. Let $\mathfrak{M}'=\{\Match\in \mathfrak{M}: \Match \textrm{ causes }|\Psi| \textrm{ to increase by one}\}$ and $\mathfrak{M}''=\mathfrak{M}\setminus \mathfrak{M}'$. Since any $\Match\in \mathfrak{M}''$ does not change $|\Psi|$ and $|\Psi|\le k-1$, we have $|\mathfrak{M}'|\le k-1$, which in turn implies $|\mathfrak{M}''|\ge k+1$. There must be a point $v\in V$ which initiates at least two external matches in $\mathfrak{M}''$. Let $\Match\in \mathfrak{M}''$ be the second external match initiated by $v$. Obviously, $\Match$ is a native match of $\pi$. Now we proceed case by case. 

\textbf{Case 1}: $v\in \Psi$ during $\Phase(\Match)$. If $v$ is aligned in $\Phase(\Match)$, by Lemma \ref{lemma:aligned initiation}, we have $round\_cost_{time}(\pi,\OPT) \ge \Distance$ . Otherwise, either $round\_cost_{time}(\pi,\OPT) \ge \Distance$ or $round\_cost_{space}(\pi,\OPT) \ge \Distance$ by Lemma \ref{lemma:no misaligned points in Psi}.

\textbf{Case 2}: $v\notin \Psi$ during $\Phase(\Match)$. Assume $\Match=\langle\rho,\rho'\rangle$ and $v=\Location(\rho), u=\Location(\rho')$. Since $\Match\in \mathfrak{M}''$, it hold that $u\in\Psi$ when $\Match$ occurs. Applying Lemma \ref{lemma:phase included in round}, we finish the proof. \qed
\end{proof}

Up to now, we have focused on good rounds. The next lemma indicates that the cost of $\OPT$ in bad rounds can be \textit{ignored} in some sense. Here, a match/round is called bad if it is not good.

\begin{lemma}\label{lemma: bad round's bound}
The number of bad rounds of $\ALG$ is at most twice the number of external matches of $\OPT$. 
\end{lemma}
\begin{proof}
The only event that causes a match of $\ALG$ to be bad is an external match of $\OPT$. 
An external match of $\OPT$ changes the alignment status of at most two points, hence leading to at most two bad  match of $\ALG$, which in turn incurs at most two bad rounds. \qed
\end{proof}

For any $v\in V$ and $t\in \mathbb{R}^+$, define $z_{v,t}$ to be the value of $z_v$ at time $t$. Then we begin to handle the cost of $\OPT$ when $z_{v,t}$ is too big, namely, on the time set $\Gamma=\{t:  \mathcal{A} \textrm{ has a pending request } \rho \textrm{ at time } t\textrm{ and } z_{\Location(\rho),t}>2\Distance\}.$ More notation is needed.

For any $v\in V$ and $t\in \mathbb{R}^+$, we say that $t$ is critical time of $v$ if $t$ is the time when the the last match involving $v$ occurs or when an external match involving $v$ occurs. Let $\Lambda=\{(v,t)\in V\times \mathbb{R}^+: t \textrm{ is critical time of } v\}$. For any $(v,t)\in \Lambda$, define $\Gamma_{v,t}=\emptyset$ be the emptyset if $z_{v,t}<2\Distance$, otherwise $\Gamma_{v,t}=\{t'<t: v\textrm{ has a pending request at time }t' \textrm{ and } z_{v,s}>2\Distance \textrm{ for }t'\le s\le t\}$.

We have the following easy observations:
\begin{itemize}
\item All the $\Gamma_{v,t}$'s are pairwise disjoint and $\Gamma=\bigcup_{(v,t)\in \Lambda}\Gamma_{v,t}$.
\item The size of $\Lambda$ is at most $k$ plus twice the number of external matches of $\ALG$.
\item Since the timers $z_v$ are reset to $0$ if and only external matches involving $v$ occurs, $\tCost_{\A}(R)=\sum_{(v,t)\in \Lambda} z_{v,t}$.
\end{itemize}

For any $x\in \mathbb{R}^+$, define its \emph{truncated value} with respect to $y\ge 0$ to be $$trun(x;y)=
\begin{cases}
0& \textrm{if } x\le y\\
f(f^{-1}(x)-f^{-1}(y))& \textrm{otherwise}
\end{cases}.$$ 

Now we derive a lower bound of the time cost of $\OPT$ on every $\Gamma_{v,t}$. 

\begin{lemma}\label{phasecostlowerbound}
For any $(v,s)\in \Lambda$, $$\sum_{u\in V}C_{time}(u,\Gamma_{v,s},\OPT)\ge trun(z_{v,s};2\Distance).$$ 
\end{lemma}

\begin{proof}
Basically, the proof is similar to that of Lemma \ref{lemma:phase included in round}. It is enough to consider the case where $\Gamma_{v,s}\neq \emptyset$.

Suppose that $\Gamma_{v,s}$ consists of disjoint intervals $I_i=(a_i,b_i]$ for $1\le i\le n$, and $b_i< a_{i+1}$ for $1\le i< n$. Then there are pending requests $\rho_1,\cdots,\rho_n$ at point $v$ such that 
\begin{itemize}
\item $T(\rho_i)=b_{i}$ for $1\le i\le n$, $t(\rho_1)\le a_1$, $t(\rho_i)=a_i$ for $1< i\le n$, 
\item $\sum_{i=1}^n c_i=z_{v,s}-2\Distance$, where $c_i=\int_{I_i}f'(t-t(\rho_i))dt$ for $1\le i\le n$.
\end{itemize}

At any time $t\in I_i$, $\ALG$ has no pending requests at points other than $v$, meaning that totally an odd number of requests have arrived by time $t$. Since a match consumes two requests, $\OPT$ must also have pending requests throughout each time interval $I_i$. Furthermore, note that no requests arrive during $I_i$. Hence, for each $1\le i\le n$, $\OPT$ has a request $\rho'_i$ at some $u_i$ that is pending throughout $I_i$. Considering that $f'$ is increasing, we have 
$$C_{time}(u_i,I_i,\OPT)\ge\int_{I_i}f'(t-t(\rho'_i))dt\ge \int_{I_i}f'(t-a_i)dt=f(b_i-a_i).$$

Therefore, 
\begin{align*}
\sum_{u\in V}C_{time}(u,\Gamma_{v,j},\OPT)&\ge \sum_{i=1}^n C_{time}(u_i,I_i,\OPT)\\
									 &\ge \sum_{i=1}^n f(b_i-a_i)\\
									 &\ge trun(z_{v,s};2\Distance). \quad\textrm{(by Lemma \ref{lemma:alignedlowerbound})}
\end{align*}
The lemma thus holds. \qed
\end{proof}

We are ready to obtain a lower bound of the time cost of $\OPT$.

\begin{lemma}\label{totalcostlowerbound}
$\tCost_{\OPT}(R)\ge \sum_{\pi\in\Pi}round\_cost_{time}(\pi,\OPT)+\sum_{(v,s)\in\Lambda} trun(z_{v,s};2\Distance)$. 
\end{lemma}
\begin{proof}
It is easy to see that
\begin{align*}
\tCost_{\OPT}(R)&\ge \sum_{\pi\in\Pi}\sum_{u\in V}C_{time}(u,\pi,\OPT)\\
&\ge \sum_{\pi\in\Pi}\sum_{u\in V}\left(C_{time}(u,\pi\setminus\Gamma,\OPT)+C_{time}(u,\pi\bigcap\Gamma,\OPT)\right)\\
&\ge \sum_{\pi\in\Pi}round\_cost_{time}(\pi,\OPT)+\sum_{u\in V}C_{time}(u,\Gamma,\OPT)
\end{align*}
Furthermore, 
\begin{align*}
\sum_{u\in V}(C_{time}(u,\Gamma,\OPT)&= \sum_{u\in V}C_{time}(u,\bigcup_{(v,s)\in\Lambda}\Gamma_{v,s},\OPT)\\
&= \sum_{u\in V}\sum_{(v,s)\in\Lambda}C_{time}(u,\Gamma_{v,s},\OPT)\\
&=\sum_{(v,s)\in\Lambda}\sum_{u\in V}C_{time}(u,\Gamma_{v,s},\OPT)\\
&\ge \sum_{(v,s)\in\Lambda}trun(z_{v,s};2\Distance)\qquad\textrm{(by Lemma \ref{phasecostlowerbound})}
\end{align*} \qed
%
%
%
\end{proof}

The following technical lemmas will be needed.

\begin{lemma}\label{ratioupperbound}
Given $c,c_0,c_1,\cdots,c_n\in \mathbb{R}^+$ with $c_i\ge c_0>c$ for any $1\le i\le n$, we have $$\frac{\sum_{j=1}^n (c_j-c)}{\sum_{j=1}^n trun(c_j;c)}\le \frac{c_0-c}{trun(c_0;c)}.$$
\end{lemma}
\begin{proof}
Recall that $f(t)=t^\alpha$ and $\alpha>1$. It suffices to prove that $\frac{a-b}{(\sqrt[\alpha]{a}-\sqrt[\alpha]{b})^{\alpha}}$ decreases with $a$ when $a>b$. This is equivalent to showing $g(x)= \frac{x^{\alpha}-y^{\alpha}}{(x-y)^{\alpha}}$ decrease with $x$ when $x>y$. The claim holds since $g'(x) = \alpha \cdot \frac{y(y^{\alpha-1}-x^{\alpha-1})}{(x-y)^{\alpha+1}} \le 0$. \qed
\end{proof}
\begin{lemma}\label{singleround}
If $\ALG$ has only one round on input instance $R$, $\frac{\Cost_{\A}(R)}{\Cost_{\adv{\A}}(R)}=O(k)$.
\end{lemma}
\begin{proof}
Denote the round by $\Round$. We proceed case by case.

\textbf{Case 1}: Both $\ALG$ and $\OPT$ have no external matches. Then they must behave in the same way on $R$. Hence $\Cost_{\A}(R)/\Cost_{\adv{\A}}(R)=1$.

\textbf{Case 2}: Either $\ALG$ or $\OPT$ has an external match. If $\OPT$ has no external matches, the first external match of $\ALG$ must be a good native match, so $round\_cost_{time}(\pi,\OPT)\ge \Distance$ by Lemma \ref{lemma:aligned initiation}. If $\OPT$ has an external match, $\sCost_{\adv{\A}}(R)\ge \Distance$. So, we always have $\sCost_{\adv{\A}}(R)+round\_cost_{time}(\pi,\OPT)\ge \Distance$.

On the one hand, $\A$ has at most $2k$ external matches in a round, so $\Cost_{\A}(R)\le 2k\Distance+\sum_{(v,s)\in\Lambda}z_{v,s}$. Let $\Lambda'=\{(v,s)\in\Lambda: z_{v,s}>4\Distance\}$. Because $|\Phi|\le 5k$, it holds that $\Cost_{\A}(R)\le 22k\Distance+\sum_{(v,s)\in\Lambda'}(z_{v,s}-2\Distance)$.

On the other hand, as to the cost of $\OPT$, we have 
\begin{align*}
\Cost_{\adv{\A}}(R)&=\sCost_{\adv{\A}}(R)+\tCost_{\adv{\A}}(R)\\
			     &\ge \sCost_{\adv{\A}}(R)+round\_cost_{time}(\pi,\OPT)+\sum_{(v,s)\in\Lambda} trun(z_{v,s};2\Distance)\\
			     &\ge \Distance+\sum_{(v,s)\in\Lambda'} trun(z_{v,s};2\Distance),
\end{align*}
where the first inequality follows from Lemma \ref{totalcostlowerbound}.

By Lemma \ref{ratioupperbound}, $\frac{\Cost_{\A}(R)}{\Cost_{\adv{\A}}(R)}\le \frac{22k\Distance+\sum_{(v,s)\in\Lambda'}(z_{v,s}-2\Distance)}{\Distance+\sum_{(v,s)\in\Lambda'} trun(z_{v,s};2\Distance)}=O(k)$. \qed
\end{proof}

Now we are ready to prove the first main result.
\CROfTrivialMetric*

\begin{proof}
Suppose that $\ALG$ has $m$ rounds on the online input instance $R$, namely $|\Pi|=m$. By Lemma \ref{singleround}, we assume that $m>1$ without loss of generality. 

Because $\ALG$ has at most $2mk$ external matches, we have $|\Lambda|\le k+2\times 2mk<5mk$. Let $\Lambda'=\{(v,s)\in\Lambda: z_{v,s}>4\Distance\}$. Then it holds that 
\begin{align*}
\Cost_{\A}(R)&= \sCost_{\A}(R)+\tCost_{\A}(R)\\
			     &\le 2km\Distance+\sum_{(v,s)\in\Lambda}z_{v,s}\\
			     &\le 22km\Distance+\sum_{(v,s)\in\Lambda'}(z_{v,s}-4\Distance)\\
			     &\le 22km\Distance+\sum_{(v,s)\in\Lambda'}(z_{v,s}-2\Distance).
\end{align*}

On the other hand, as for the cost of $\OPT$, we have  by Lemma \ref{totalcostlowerbound} that 
\begin{align*}
\Cost_{\adv{\A}}(R)&=\sCost_{\adv{\A}}(R)+\tCost_{\adv{\A}}(R)\\
			     &\ge \sCost_{\adv{\A}}(R)+\sum_{\Round\in\Pi} round\_cost_{time}(\pi,\OPT)+\sum_{(v,s)\in\Lambda} trun(z_{v,s};2\Distance)\\
			     &=\sum_{\Round\in\Pi}(round\_cost_{space}(\pi,\OPT)+round\_cost_{time}(\pi,\OPT))\\
			     &\quad +\sum_{(v,s)\in\Lambda'} trun(z_{v,s};2\Distance).
\end{align*}
Let $\Pi'$ be the set of good complete rounds and $m'=|\Pi'|$. Let $m''$ be the number of bad rounds. An easy observation is that $m'+m''\ge m-1$. By Lemma \ref{lemma: bad round's bound}, $\OPT$ has at least $\frac{m''}{2}$ external matches, meaning that $\Cost_{\adv{\A}}(R)\ge \sCost_{\adv{\A}}(R)\ge \frac{m''}{2}\Distance$. As a result,
\begin{align*}
2\Cost_{\adv{\A}}(R)&\ge\sum_{\Round\in\Pi}(round\_cost_{space}(\pi,\OPT)+round\_cost_{time}(\pi,\OPT))\\
			     &\quad +\sum_{(v,s)\in\Lambda'} trun(z_{v,s})+\frac{m''}{2}\Distance\\
			     &\ge f(f^{-1}(2\Distance)-f^{-1}(\Distance)) m'+\frac{m''}{2}\Distance+\sum_{(v,s)\in\Lambda'}  trun(z_{v,s};2\Distance)\\
			     &\ge\frac{m-1}{2}(\sqrt[\alpha]{2}-1)^\alpha \Distance+\sum_{(v,s)\in\Lambda'} trun(z_{v,s};2\Distance)\
\end{align*}
where the second inequality is due to Lemma \ref{lemma: at most n exchanges}.

By Lemma \ref{ratioupperbound}, $\frac{\Cost_{\A}(R)}{\Cost_{\adv{\A}}(R)}\le \frac{22km\Distance+\sum_{(v,s)\in\Lambda'}(z_{v,s}-2\Distance)}{\frac{m-1}{4}(\sqrt[\alpha]{2}-1)^\alpha \Distance+\frac{1}{2}\sum_{(v,s)\in\Lambda'} trun(z_{v,s};2\Distance)}=O(k)$. \qed
%
%
%
%

\end{proof}

\section{Lower Bound for Deterministic Algorithms}\label{deterministiclowerbound}

This section is devoted to showing that any deterministic algorithm for the convex-MPMD problem on $k$-point uniform metric space must have competitive ratio $\Omega(k)$, meaning that our algorithm is optimum up to a constant factor.
Let's begin with a convention of notation. Let $f:\mathbb{R}^+\mapsto\mathbb{R}^+$ be a nondecreasing, unbounded, continuous function satisfying $f(0)=f'(0)=0$.  Let $\Space = (V,\Distance)$ be a uniform metric space with $V=\{v_0,v_1,...v_k\}$. 
Suppose that $\ALG$ is an arbitrary deterministic online algorithm for the $f$-MPMD problem. Let $T\in\mathbb{R}^+$ be such that $f(T)=k\Distance$. Arbitrarily fix a real number $\tau>0$ such that $n=\frac{T}{\tau}$ is an even number. 

We construct an instance $R$ of online input to $\ALG$ and show that the competitive ratio of $\ALG$ is at least $\Omega(k)$. The instance $R$ is determined in an online fashion: Roughly speaking, based on the up-to-now behavior of $\ALG$, we decide the time and locations of the following requests so that $\ALG$ has to take many external matches. 

Specifically, $R$ is constructed in $m$ rounds, where $m$ is an arbitrary positive integer. The first round begins at time $T_1=0$. Some requests arrive in the manner as described in the next two paragraphs. At arbitrary time $T_2$ after these requests are all matched, finish the first round and start the second round. Repeat this process until we have finished $m$ rounds. All the requests form the instance $R$.

Now fix $1\le r\le m$, and we describe the requests that arrive during the $r$th round, namely in the interval $[T_r,T_{r+1})$. Define $G_0=(V,\emptyset)$ to be the edgeless graph on $V$, and let $C_0=\{v_0\}, h=1$. At time $T_r$, a request $\rho_{00}$ arrives at $v_0$. The other requests are specified as follows.

For any integers $(h-1)n< j\le hn$ and $1\le i\le k$ such that $v_i\in V\setminus C_{h-1}$, a request $\rho_{ij}$ arrives at point $v_i$ at time $T_r+j\tau$. When $\rho_{i,hn}$ arrives, namely, at time $T_r+hT$, 
construct an undirected graph $G_h$ on $V$. It has an edge between any pair of vertices $v_i\neq v_{i'}$ if and only if by time $T_r+hT$, $\ALG$ has matched one request at $v_i$ and another at $v_{i'}$ both arriving in the period $[T_r, T_r+hT]$. Let $C_h$ be the set of the vertices in the connected component of $G_h$ that contains $v_0$. If $C_{h-1}\neq C_h\neq V$, increase $h$ by 1 and iterate the process described in this paragraph. Otherwise, arbitrarily choose $i$ such that $v_i\in V\setminus C_{h-1}$, let one more request $\rho_{i,hn+1}$ arrive at point $v_i$ at time $T_r+hT+\tau$, and finish the construction of round $r$; denote the final $h$ by $h_r$.

Arbitrarily fix $1\le r\le m$ in the rest of this section. 
Let  $R_r$ be the set of requests that arrive in the first $r$ rounds, and $N_r$ be the number of requests in $R_r\setminus R_{r-1}$, where $R_0=\emptyset$. Let $R=R_m$. It is easy to see four facts:
\begin{description}
\item[Fact 1:] $N_r\le k^2 n+2$.
\item[Fact 2:] $R_r\setminus R_{r-1}$ has exactly one request at $v_0$, and has an odd number of requests at the point where the last request in round $r$ arrives.
\item[Fact 3:] $R_r\setminus R_{r-1}$ has an even number of requests at any other point.
\item[Fact 4:] No match occurs between requests of different rounds.\label{localmatches}
\end{description}

Some lemmas are needed for proving the main result of this section.

\begin{lemma}\label{optupperbound} 
$\Cost_{\adv{\A}}(R_r)\le (\Distance+\frac{k^2 n}{2}f(\tau)+f(\tau))r$. 
\end{lemma}
\begin{proof}
It suffices to show that the cost that $\OPT$ pays for any round is at most $\Distance+\frac{k^2 n}{2}f(\tau)+f(\tau)$. Without loss of generality, we prove this for the first round and assume that the last request of this round is located at $v_k$. By Facts 2 and 3, the requests of this round can be paired up in this way: $\langle\rho_{00},\rho_{k1}\rangle$, $\langle\rho_{ij},\rho_{i,j+1}\rangle$ for $1\le i\le k-1$ and odd number $j\ge 1$, and $\langle\rho_{kj},\rho_{k,j+1}\rangle$ for even numbers $j\ge 2$. Since $\adv{\A}$ is an optimum offline algorithm, its cost is at most the cost of this matching. \qed
\end{proof}

\begin{lemma}\label{alglowerbound}
$\Cost_{\A}(R_r)\ge k\Distance r$.
\end{lemma}
\begin{proof}
By Fact 4, it is equivalent to show that the cost that $\ALG$ pays for requests in $R_r\setminus R_{r-1}$ is at least $k\Distance$. We proceed case by case.

\textbf{Case 1:} $C_{h_r}=C_{h_r -1}$.
On the one hand, assume Case 2 in this round does happen. We have three observations:
\begin{itemize}
\item After time $T_r+(h_r-1)T$, no request arrives at any $v\in C_{h_r}$.
\item The total number of requests that have arrived at points in $C_{h_r}$ is an odd number. Hence, there must be a request $\rho$ such that (1) $\Location(\rho)\in C_{h_r}$ and (2) $\ALG$ eventually matches $\rho$ with another request $\rho'$ satisfying $\Location(\rho')\notin C_{h_r}$. 
\item Since no external match occurs between $C_{h_r}$ and $V\setminus C_{h_r}$ in period $(T_r+(h_r-1)T, T_r+hT]$, the request $\rho$ must be pending throughout the period, incurring time cost at least $f(T)=k\Distance$.
\end{itemize}

\textbf{Case 2:} $C_{h_r}=V$. Then $\ALG$ has at least $k$ external matches in this round. 

Altogether, the cost $\ALG$ pays for this round is at least  $k\Distance$. \qed
\end{proof}

\UniformMetricLowerBound*

\begin{proof}
Suppose there are $a=a(k,\Distance)$ and $b=b(k,\Distance)$ such that for any $m\ge 1$, 
$\Cost_{\A}(R)\le a\cdot \Cost_{\adv{\A}}(R)+b.$
Fix $k$ and $\Distance$. Dividing both sides of inequality by $m$ and letting $m$ approach infinity, by Lemmas \ref{optupperbound} and \ref{alglowerbound}, we get $f(n\tau)\le (\Distance+\frac{k^2 n}{2}f(\tau)+f(\tau))a$, which means that $a\ge \frac{f(n\tau)}{\Distance+\frac{k^2 n}{2}f(\tau)+f(\tau)}=\frac{\frac{k\Distance}{2}+\frac{1}{2}f(n\tau)}{\Distance+\frac{k^2 n}{2}f(\tau)+f(\tau)}$.

Let $\tau$ approach zero. One has $\lim_{\tau\rightarrow 0}f(\tau)=0$, and 
\begin{align*}
\lim_{\tau\rightarrow 0}\frac{f(n\tau)}{k^2 nf(\tau)}=\lim_{\tau\rightarrow 0}\frac{1}{k^2}\frac{f(n\tau)}{n\tau}\frac{\tau}{f(\tau)}
 =\lim_{\tau\rightarrow 0}\frac{1}{k^2}\frac{f(T)}{T}\frac{\tau}{f(\tau)}
 =+\infty \quad \textrm{since }f'(0)=0
\end{align*} 
This means $\lim_{\tau\rightarrow 0}k^2 nf(\tau)=0$, since $f(n\tau)=k\Distance$ is a constant when $k$ and $\Distance$ are fixed. As a result, $a=\lim_{\tau\rightarrow 0}a\ge \lim_{\tau\rightarrow 0}\frac{\frac{k\Distance}{2}+\frac{1}{2}f(n\tau)}{\Distance+\frac{k^2 n}{2}f(\tau)+f(\tau)}=\frac{k\Distance}{\Distance}=k$. \qed
\end{proof}

\section{Lower Bound for Randomized Algorithms}
As in Section \ref{deterministiclowerbound}, arbitrarily fix a nondecreasing, unbounded, continuous function  $f:\mathbb{R}^+\mapsto\mathbb{R}^+$ which satisfies $f(0)=f'(0)=0$.  Let $\Space = (V,\Distance)$ be a uniform metric space with $V=\{v_0,v_1,...v_k\}$. This section is devoted to show that against oblivious adversaries, the competitive ratio of any randomized algorithm for $f$-MPMD is lower bounded by $\Omega(\ln k)$. 

By Yao's principle, it is enough to show that there is a randomized input instance on $V$ such that any \emph{deterministic} $f$-MPMD algorithm must have expected competitive ratio $\Omega(\ln k)$. 

To construct the random instance, let $T\in\mathbb{R}^+$ be such that $f(T)= k!\Distance\ln k$. Arbitrarily choose real number $\tau>0$ such that $n=\frac{T}{\tau}$ is an even number. 
Let $R$ be an instance of online input over $\Space$  such that
\begin{itemize}
\item Point $v_0$ has one request, which arrives at time $t=0$;
\item For $1\le i\le k-1$, each point $v_i$ has $ni$ requests, which arrive at time $j\tau$, $1\le j\le ni$, respectively;
\item Point $v_k$ has $kn+1$ requests, which arrive at time $j\tau$, $1\le j\le kn+1$, respectively.
\end{itemize}

An example of $R$ with $k=3$ is illustrated in Figure \ref{fig:randomlowerbound}(a).
\begin{figure}
  \centering
  \includegraphics[width=\textwidth]{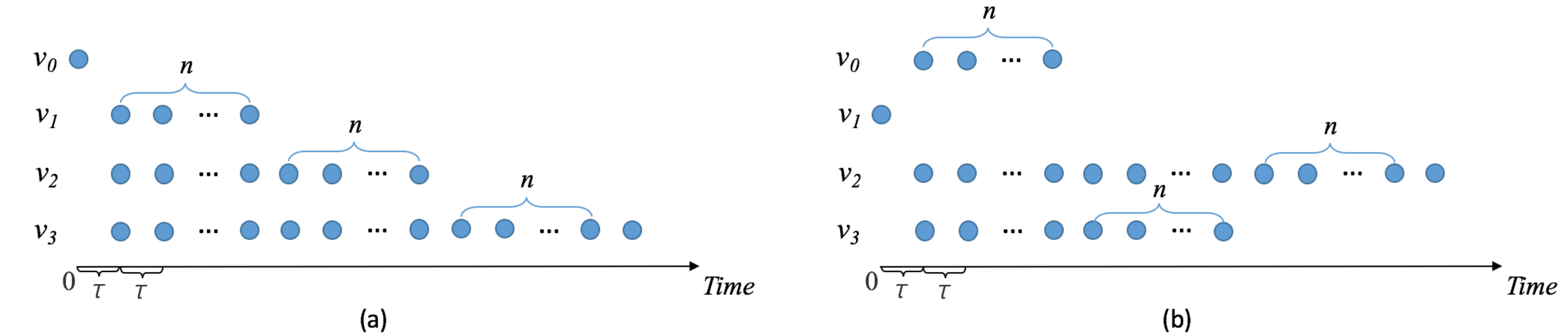}
  \caption{(a) An input instance $R$ on a 4-point uniform metric space. (b) $R_{\sigma}$ for the permutation $\sigma: 0\mapsto 1,1\mapsto 0,2\mapsto 3,3\mapsto 2$.}\label{fig:randomlowerbound}
\end{figure}

Let $\Xi_k$ be the set of permutations on $\{0,1,...,k\}$. For any permutation $\sigma\in \Xi_k$, define $R_{\sigma}$ to be the online input instance in which each point $v_{\sigma(i)}$ has the same requests as $v_i$ does in $R$, $0\le i\le k$. An example of $R_{\sigma}$ is illustrated in Figure \ref{fig:randomlowerbound}(b).

Let $\bm{\sigma}$ be a random permutation uniformly distributed on $\Xi_k$. Then $R_{\bm{\sigma}}$ is uniformly distributed on the set $\{R_\sigma: \sigma\in \Xi_k\}$. 

Arbitrarily fix a deterministic algorithm $\ALG$ solving $f$-MPMD problem. Let $\OPT$ be the optimum offline algorithm. The main task is to compare the expected cost of $\ALG$ with that of $\OPT$ on the input instance $R_{\bm{\sigma}}$. Let's begin with some lemmas, where $\mathbb{E}[\cdot]$ means expectation. 

\begin{lemma}\label{optlowerboundsingle}
When $k$ is fixed, $\lim_{\tau\rightarrow 0}\mathbb{E}[\Cost_{\adv{\A}}(R_{\bm{\sigma}})]\le \Distance$.
\end{lemma}
\begin{proof}
Consider the following offline matching of $R_{\bm{\sigma}}$: one external match occurs between the request at point $v_{\bm{\sigma}(0)}$ and the first request at $v_{\bm{\sigma}(k)}$, and the other matches are internal and occur between requests that arrive consecutively. The cost is $\Distance+(\frac{nk(k+1)}{4}+1)f(\tau)\le \Distance+n(k+1)^2 f(\tau)$. Because $\OPT$ is an optimum offline algorithm, it must hold that $\Cost_{\adv{\A}}(R_{\bm{\sigma}})]\le \Distance+n(k+1)^2 f(\tau)$.

Furthermore, since $k$ is fixed, $$\lim_{\tau\rightarrow 0}nf(\tau)=\lim_{\tau\rightarrow 0}n\tau\frac{f(\tau)}{\tau}=f^{-1}(k!\Distance\ln k)f'(0)=0,$$which immediately leads to the lemma. \qed
\end{proof}

Now we consider the behavior of $\ALG$ on the random input instance $R_{\bm{\sigma}}$. 

\begin{lemma}\label{alglowerboundsinglelongdelay}
If with nonzero probability $\ALG$ keeps some request pending no shorter than  $T$, $\mathbb{E}[\Cost_{\A}(R_{\bm{\sigma}})]\ge \Distance\ln k$.
\end{lemma}
\begin{proof}
Let $\sigma'\in \Xi_k$ be such that on the input instance $R_{\sigma'}$, $\ALG$ keeps a request pending for a period at least $T$. The time cost incurred by this delay is at least $f(T)=k!\Distance\ln k$, so $\Cost_{\A}(R_{\sigma'})\ge k!\Distance\ln k$. Since $\bm{\sigma}$ is uniformly distributed on $\Xi_k$, we have $\Pr(\bm{\sigma}=\sigma')=\frac{1}{k!}$, which implies that $\mathbb{E}[\Cost_{\A}(R_{\bm{\sigma}})]\ge \Pr(\bm{\sigma}=\sigma')\Cost_{\A}(R_{\sigma'})\ge \Distance\ln k$. \qed
\end{proof}

Actually Lemma \ref{alglowerboundsinglelongdelay} remains true even if no request is delayed too long, as claimed in the next lemma.
\begin{lemma}\label{alglowerboundsingleshortdelay}
If $\ALG$ doesn't keep any request pending longer than $T$, $\mathbb{E}[\Cost_{\A}(R_{\bm{\sigma}})]\ge \Distance\Theta(\ln k)$.
\end{lemma}
\begin{proof}
 The basic idea of the proof is to show that in expectation, $\ALG$ has $\ln k$ external matches. The number of external matches can be estimated by sampling the points where an externally matched request is located. This sampling is done sequentially while $\ALG$ is running. The detail is described below.

Suppose that we have already sampled $u_1,\cdots,u_{h-1}\in V$ for some $h\ge 1$. If $h=1$, let $j_0=0$. If $h>0$, wait until no more requests arrive at $u_{h-1}$, and by definition of $R_{\bm{\sigma}}$, we can determine both $j_{h-1}$ such that $u_{h-1}=v_{\bm{\sigma}(j_{h-1})}$ and $\bm{\sigma}(j)$ for $j\le j_{h-1}$. Then we proceed case by case.

%
\textbf{Case 1: } $j_{h-1}=k$. The sampling process ends and output $\{u_1,...u_{h-1}\}$ as the final sample. Let $h_{\bm{\sigma}}=h-1$.

\textbf{Case 2: } $j_{h-1}<k$. At time $(j_{h-1}+1)T$, let $$V_h=\{v_{\bm{\sigma}(j)}: j\le j_{h-1} \} \textrm{ and}$$  $$U_h=\{u\in V\setminus V_h: \textrm{there is  a match } \langle\rho,\rho'\rangle \textrm{ with } l(\rho)=u,l(\rho')\in V_h\}.$$
Define $\mathcal{R}_h=\{\rho\in R_{\bm{\sigma}}: \Location(\rho)\in V_h \}$. Note that all the requests in $\mathcal{R}_h$ have arrived by time $j_{h-1}T$, hence being matched by time $(j_{h-1}+1)T$ according to the assumption that $\ALG$ doesn't keep any request pending longer than $T$. Because $j_{h-1}<k$, $\mathcal{R}_h$ has an odd number of requests, implying that by time $(j_{h-1}+1)T$, $\ALG$ must have matched one request in $\mathcal{R}_h$ with another located outside $V_h$. As a result, $U_h\neq\emptyset$. Arbitrarily choose $u_h\in U_h$, and iterate the process with already-sampled points $u_1,\cdots,u_{h}$. 


The above sampling process finally produce a sample $\{u_1,\cdots, u_{h_{\bm{\sigma}}}\}$. Recall the indices $j_i$ satisfies $v_{\bm{\sigma}(j_i)}=u_i$, $1\le i\le h_{\bm{\sigma}}$. Due to the randomness of $\bm{\sigma}$, all the indices $j_i$ and $h_{\bm{\sigma}}$ are random variables. It is easy to observe two facts.

\begin{description}

\item[\textbf{Fact 1}:] for any $1\le h\le h_{\bm{\sigma}}$, given $j_{h-1}$, $j_h$ is uniformly distributed on the set $\{j_{h-1}+1,...,k\}$. This follows from the uniformly distribution of $R_{\bm{\sigma}}$. 

\item[\textbf{Fact 2}:] $\ALG$ has at least $h_{\bm{\sigma}}$ external matches. This obviously holds.
\end{description}

Now we estimate $\mathbb{E}[h_{\bm{\sigma}}]$. For any integer $1\le i\le k$, let $S_i$ stands for the set $\{i,\cdots,k\}$, and define a random variable $X_i$ by induction as follows: 
\begin{itemize}
\item Uniformly randomly sample $c_1$ from $S_i$, 
\item Suppose $c_{h-1}$ has been sampled. If $c_{h-1}<k$, uniformly randomly sample $c_h$ from $S_{c_{h-1}+1}$, and iterate this step with $h$ increased  by 1. Otherwise, the induction ends and let $X_i=h-1$.
\end{itemize}

We can see that for any $1\le i\le k$, $\mathbb{E}[X_i]=1+\frac{1}{k-i+1}\sum_{j=i+1}^k \mathbb{E}[X_j]$. Considering that $\mathbb{E}[X_k]=1$, we have $\mathbb{E}[X_1]=\sum_{i=1}^k \frac{1}{i}$.  

By Fact 1 mentioned above, $h_{\bm{\sigma}}$ and $X_1$ are identically distributed. As a result, $\mathbb{E}[h_{\bm{\sigma}}]=\mathbb{E}[X_1]=\Theta(\ln k)$.

Since an external match incurs space cost $\Distance$, the proof ends due to Fact 2. \qed
\end{proof}
%
%
%
%
%

For any $t\in \mathbb{R}^+$, define an online input instance $R^{(t)}$ which is derived from $R$ by deferring the arrival of each request for a period of length $t$. Likewise define $R^{(t)}_\sigma$ for any $\sigma\in \Xi_k$. Arbitrarily fix a positive integer $m$. For any $1\le r\le m$, Let $\bm{\sigma}_1,...,\bm{\sigma}_m$ be independent random variables that are uniformly distributed on $\Xi_k$. Our ultimate random online input instance is $\bm{R}=\bigcup_{r=1}^m R^{(T_r)}_{\bm{\sigma}_r}$, where $T_1=0$ and  for any $1\le r\le m$, $T_r$ is arbitrary time before which all requests in $\bigcup_{i=1}^{r-1} R^{(T_i)}_{\bm{\sigma}_i}$ have been matched. 
The instance $\bm{R}$ is illustrated in Figure \ref{fig:obliviousadverary}.

\begin{figure}
  \centering
  \includegraphics[width=\textwidth]{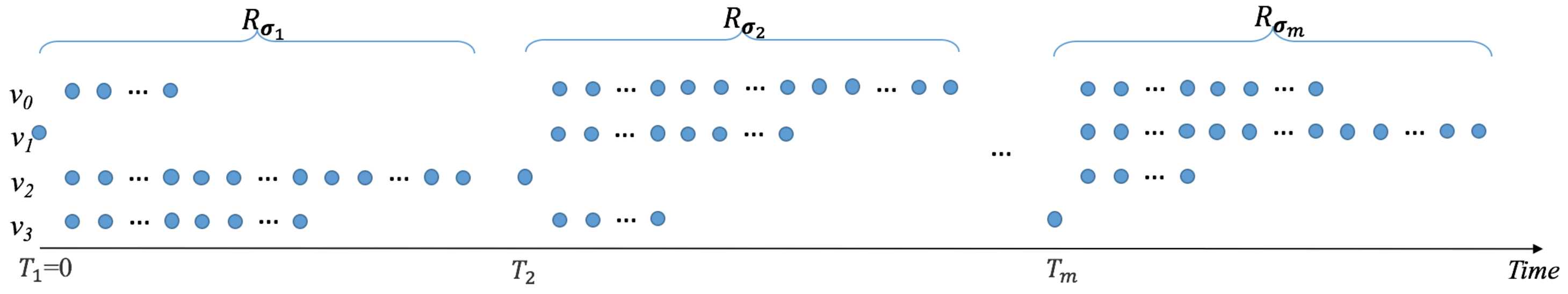}
  \caption{An illustration of $\bm{R}$ on a 4-point uniform metric space.}\label{fig:obliviousadverary}
\end{figure}
 
By the definition of $\bm{R}$, we know that any pair of requests paired up by $\ALG$ must belongs to the same $R^{(T_r)}_{\bm{\sigma}_r}$, for some $1\le r\le m$. Hence, $\Cost_{\adv{\A}}(\bm{R})\le \sum_{r=1}^m \Cost_{\adv{\A}}(R^{(T_r)}_{\bm{\sigma}_r})$ and $\Cost_{\ALG}(\bm{R})= \sum_{r=1}^m \Cost_{\ALG}(R^{(T_r)}_{\bm{\sigma}_r})$. This, together with the trivial observation that Lemmas \ref{optlowerboundsingle}-\ref{alglowerboundsingleshortdelay} remain true for any instance $R^{(T_r)}_{\bm{\sigma}_r}$, immediately leads to the following lemma.

\begin{lemma}\label{algoptbounds}
$\mathbb{E}[\Cost_{\A}(\bm{R})]\ge m\Distance\Theta(\ln k)$, and $\lim_{\tau\rightarrow 0}\mathbb{E}[\Cost_{\adv{\A}}(\bm{R})]\le m\Distance$ when $k$ is fixed.
\end{lemma}

Now we are ready to prove the main result of this section.

\begin{theorem}
Suppose that the time cost function $f$ is nondecreasing, unbounded, continuous and satisfies $f(0)=f'(0)=0$. 
Then any randomized algorithm for $f$-MPMD on $k$-point uniform metric space has competitive ratio $\Omega(\ln k)$ against oblivious adversaries. 
\end{theorem}
\begin{proof}
Arbitrarily choose a deterministic online algorithm $\ALG$ for the $f$-MPMD problem. Let $\OPT$ be an optimum offline algorithm for the $f$-MPMD problem. Consider the random online input instance $\bm{R}$ constructed above.

Suppose there are $a=a(k,\Distance)$ and $b=b(k,\Distance)$ such that for any $m\ge 1$, 
$$\mathbb{E}[\Cost_{\A}(\bm{R})]\le a\cdot \mathbb{E}[\Cost_{\adv{\A}}(\bm{R})]+b.$$
Fix $k$ and $\Distance$. Letting $\tau$ approach zero, by Lemma \ref{algoptbounds}, we have $m\Distance\Omega(\ln k)\le am\Distance+b$. Dividing both sides of the inequality by $m$ and letting $m$ approach infinity, we get $a= \Omega(\ln k)$. The lemma is proven due to Yao's principle. \qed
\end{proof}

\section{Conclusion} 
\label{Sec:Disscution}
We have formulate the convex-MPMD problem and have designed a deterministic online algorithm to solve this problem. For any monomial function $f(t)=t^\alpha$ with $\alpha>1$, this algorithm has competitive ratio $O(k)$ if requests are from a $k$-point uniform metric space. Actually, this result can be generalized to any convex function that satisfies (a) $f(0)=f'(0)=0$, (b) $\ln f(t)$ is a concave function, and (c) for any $\lambda_1<\lambda_2\in \mathbb{R^+}$, $\sup_{x\in \mathbb{R^+}}\frac{x}{f(f^{-1}(\lambda_2 x)-f^{-1}(\lambda_1 x))}<\infty$. Under these conditions, all the lemmas remain true and the proof of Theorem \ref{CROfTrivialMetric} remains valid.

Our algorithm can be adapted to work on general metric spaces. Given a metric space $\Space$, let $d_{\max}$ (respectively, $d_{\min}$) stand for the maximum (respectively, minimum) pairwise distance among $\Space$. To handle requests from $\Space$, replaced $\Distance$ in the algorithm by $d_{\max}$. All the lemmas in Section \ref{sec:alg} remain true up to the slight modifications: $\Distance$ in the lower bounds of the space cost of $\OPT$ is replaced by $d_{\min}$, and $\Distance$ in the upper bounds of of the space cost $\ALG$ is replaced by $d_{\max}$. As a result, Theorem \ref{CROfTrivialMetric} still holds, if the upper bound $O(k)$ is relaxed to $O(k\Delta)$, where $\Delta\triangleq \frac{d_{\max}}{d_{\min}}$.

We have also proved lower bound $\Omega( k)$ (respectively, $\Omega(\log k)$) of the competitive ratio of any deterministic (respectively, randomized) algorithm solving convex-MPMD. This means that our deterministic algorithm is the optimum while optimum randomized algorithms remain open.

Hence, an interesting future direction is to design a randomized algorithm for convex-MPMD. A randomized algorithm is usually more competitive than a deterministic one when considering oblivious adversaries. We conjecture that there is a randomized algorithm for convex-MPMD with competitive ratio $O(\log k)$. If this turns out true, there is still a clear separation between linear-MPMD and convex-MPMD in the context of randomized algorithms.

In contrast to convex functions, concave functions may model the fact that in some applications the delay cost grows slower and slower, which encourages matching two new requests instead of matching old requests. It seems not difficult to design an algorithm with bounded competitive ratio for concave cost functions, but to design a good one, i.e., with a very small competitive ratio, seems still challenging. 


\begin{acknowledgements}
This work is partially supported by the National Key Research and Development Program of China (Grant No. 2016YFB1000200), the National Natural Science Foundation of China (61420106013), State Key Laboratory of Computer Architecture Open Fund (CARCH3410).
\end{acknowledgements}

\bibliographystyle{abbrv}
\bibliography{ref}

\begin{thebibliography}{10}

\bibitem{AggarwalGKM2011}
G.~Aggarwal, G.~Goel, C.~Karande, and A.~Mehta.
\newblock {Online Vertex-Weighted Bipartite Matching and Single-bid Budgeted
  Allocations}.
\newblock In {\em {Proceedings of the Twenty-Second Annual {ACM-SIAM} Symposium
  on Discrete Algorithms, {SODA}}}, page 1253–1264, 2011.

\bibitem{AACCGKMWW2017}
I.~Ashlagi, Y.~Azar, M.~Charikar, A.~Chiplunkar, O.~Geri, H.~Kaplan,
  R.~Makhijani, Y.~Wang, and R.~Wattenhofer.
\newblock {Min-cost Bipartite Perfect Matching with Delays}.
\newblock In {\em {20th International Workshop on Approximation Algorithms for
  Combinatorial Optimization Problems (APPROX), Berkeley, California, USA}},
  August 2017.

\bibitem{AzarCK2017arXiv}
Y.~Azar, A.~Chiplunkar, and H.~Kaplan.
\newblock {Polylogarithmic Bounds on the Competitiveness of Min-cost Perfect
  Matching with Delays}.
\newblock {\em abs/1610.05155}, 2016.

\bibitem{AzarCK2017}
Y.~Azar, A.~Chiplunkar, and H.~Kaplan.
\newblock {Polylogarithmic Bounds on the Competitiveness of Min-cost Perfect
  Matching with Delays}.
\newblock In {\em {Proceedings of the Twenty-Eighth Annual ACM-SIAM Symposium
  on Discrete Algorithms, SODA}}, page 1051–1061, 2017.

\bibitem{azar2019price}
Y.~Azar, Y.~Emek, R.~van Stee, and D.~Vainstein.
\newblock The price of clustering in bin-packing with applications to
  bin-packingwith delays.
\newblock In {\em The 31st ACM Symposium on Parallelism in Algorithms and
  Architectures}, pages 1--10, 2019.

\bibitem{azar2018deterministic}
Y.~Azar and A.~J. Fanani.
\newblock Deterministic min-cost matching with delays.
\newblock In {\em International Workshop on Approximation and Online
  Algorithms}, pages 21--35. Springer, 2018.

\bibitem{azar2017online}
Y.~Azar, A.~Ganesh, R.~Ge, and D.~Panigrahi.
\newblock Online service with delay.
\newblock In {\em Proceedings of the 49th Annual ACM SIGACT Symposium on Theory
  of Computing}, pages 551--563. ACM, 2017.

\bibitem{azar2019general}
Y.~Azar and N.~Touitou.
\newblock General framework for metric optimization problems with delay or with
  deadlines.
\newblock In {\em 2019 IEEE 60th Annual Symposium on Foundations of Computer
  Science (FOCS)}, pages 60--71. IEEE, 2019.

\bibitem{bienkowski2018primal}
M.~Bienkowski, A.~Kraska, H.-H. Liu, and P.~Schmidt.
\newblock A primal-dual online deterministic algorithm for matching with
  delays.
\newblock In {\em International Workshop on Approximation and Online
  Algorithms}, pages 51--68. Springer, 2018.

\bibitem{bienkowski2017match}
M.~Bienkowski, A.~Kraska, and P.~Schmidt.
\newblock A match in time saves nine: Deterministic online matching with
  delays.
\newblock In {\em International Workshop on Approximation and Online
  Algorithms}, pages 132--146. Springer, 2017.

\bibitem{BirnbaumM2008}
B.~E. Birnbaum and C.~Mathieu.
\newblock {On-line bipartite matching made simple}.
\newblock {\em {SIGACT} News}, 39(1):80–87, 2008.

\bibitem{DevanurJK2013}
N.~R. Devanur, K.~Jain, and R.~D. Kleinberg.
\newblock {Randomized Primal-Dual analysis of {RANKING} for Online BiPartite
  Matching}.
\newblock In {\em {Proceedings of the Twenty-Fourth Annual {ACM-SIAM} Symposium
  on Discrete Algorithms, {SODA}}}, page 101–107, 2013.

\bibitem{dooly1998tcp}
D.~R. Dooly, S.~A. Goldman, and S.~D. Scott.
\newblock Tcp dynamic acknowledgment delay (extended abstract) theory and
  practice.
\newblock In {\em Proceedings of the thirtieth annual ACM symposium on Theory
  of computing}, pages 389--398, 1998.

\bibitem{dooly2001line}
D.~R. Dooly, S.~A. Goldman, and S.~D. Scott.
\newblock On-line analysis of the tcp acknowledgment delay problem.
\newblock {\em Journal of the ACM (JACM)}, 48(2):243--273, 2001.

\bibitem{Edmonds1965b}
J.~Edmonds.
\newblock {Maximum matching and a polyhedron with 0, 1-vertices}.
\newblock {\em Journal of Research of the National Bureau of Standards B},
  69:125–130, 1965.

\bibitem{Edmonds1965a}
J.~Edmonds.
\newblock {Paths, trees, and flowers}.
\newblock {\em Canadian Journal of Mathematics}, 17:449–467, 1965.

\bibitem{EmekKW2016}
Y.~Emek, S.~Kutten, and R.~Wattenhofer.
\newblock {Online Matching: Haste makes Waste!}
\newblock In {\em {48th Annual Symposium on Theory of Computing (STOC)}}, June
  2016.

\bibitem{Emek2017Minimum}
Y.~Emek, Y.~Shapiro, and Y.~Wang.
\newblock {Minimum Cost Perfect Matching with Delays for Two Sources}.
\newblock In {\em {10th International Conference on Algorithms and Complexity
  (CIAC), Athens, Greece}}, May 2017.

\bibitem{GoelM2008}
G.~Goel and A.~Mehta.
\newblock {Online budgeted matching in random input models with applications to
  Adwords}.
\newblock In {\em {Proceedings of the Nineteenth Annual {ACM-SIAM} Symposium on
  Discrete Algorithms, {SODA}}}, page 982–991, 2008.

\bibitem{KalyanasundaramP1993}
B.~Kalyanasundaram and K.~Pruhs.
\newblock {Online Weighted Matching}.
\newblock {\em J. Algorithms}, 14(3):478–488, 1993.

\bibitem{karlin1994competitive}
A.~R. Karlin, M.~S. Manasse, L.~A. McGeoch, and S.~Owicki.
\newblock Competitive randomized algorithms for nonuniform problems.
\newblock {\em Algorithmica}, 11(6):542--571, 1994.

\bibitem{karlin1988competitive}
A.~R. Karlin, M.~S. Manasse, L.~Rudolph, and D.~D. Sleator.
\newblock Competitive snoopy caching.
\newblock {\em Algorithmica}, 3(1-4):79--119, 1988.

\bibitem{KarpVV1990}
R.~M. Karp, U.~V. Vazirani, and V.~V. Vazirani.
\newblock {An Optimal Algorithm for On-line Bipartite Matching}.
\newblock In {\em {Proceedings of the 22nd Annual {ACM} Symposium on Theory of
  Computing}}, page 352–358, 1990.

\bibitem{KhullerMV1994}
S.~Khuller, S.~G. Mitchell, and V.~V. Vazirani.
\newblock {On-Line Algorithms for Weighted Bipartite Matching and Stable
  Marriages}.
\newblock {\em Theor. Comput. Sci.}, 127(2):255–267, 1994.

\bibitem{Mehta2013}
A.~Mehta.
\newblock {Online Matching and Ad Allocation}.
\newblock {\em Foundations and Trends in Theoretical Computer Science},
  8(4):265–368, 2013.

\bibitem{MehtaSVV2005}
A.~Mehta, A.~Saberi, U.~V. Vazirani, and V.~V. Vazirani.
\newblock {AdWords and Generalized On-line Matching}.
\newblock In {\em {46th Annual {IEEE} Symposium on Foundations of Computer
  Science {(FOCS)}}}, page 264–273, 2005.

\bibitem{MeyersonNP2006}
A.~Meyerson, A.~Nanavati, and L.~J. Poplawski.
\newblock {Randomized online algorithms for minimum metric bipartite matching}.
\newblock In {\em {Proceedings of the Seventeenth Annual {ACM-SIAM} Symposium
  on Discrete Algorithms, {SODA}}}, 2006.

\bibitem{Miyazaki2014}
S.~Miyazaki.
\newblock {On the advice complexity of online bipartite matching and online
  stable marriage}.
\newblock {\em Inf. Process. Lett.}, 114(12):714–717, 2014.

\bibitem{NaorW2015}
J.~Naor and D.~Wajc.
\newblock {Near-Optimum Online Ad Allocation for Targeted Advertising}.
\newblock In {\em {Proceedings of the Sixteenth {ACM} Conference on Economics
  and Computation, {EC}}}, page 131–148, 2015.

\end{thebibliography}

\end{document}